\documentclass[11pt]{article}

\usepackage[letterpaper,margin=1in]{geometry}
\usepackage[T1]{fontenc}
\usepackage[utf8]{inputenc}
\usepackage{amsmath,amssymb,amsfonts}
\usepackage{amsthm}
\usepackage{stmaryrd}
\usepackage{booktabs}
\usepackage{array}
\usepackage{float}
\usepackage{microtype}
\usepackage{xcolor}
\usepackage{graphicx}
\usepackage{tikz}
\usetikzlibrary{arrows.meta}
\usepackage[hidelinks]{hyperref}

\newcommand{\FormCraft}{\textnormal{FormCraft}}
\newcommand{\Library}{\mathcal{L}}
\newtheorem{proposition}{Proposition}
\newtheorem{corollary}{Corollary}
\newtheorem{axiom}{Axiom}
\theoremstyle{definition}
\newtheorem{definition}{Definition}
\theoremstyle{plain}

\definecolor{formationblue}{RGB}{5,73,185}
\definecolor{borderblue}{RGB}{103,157,244}
\definecolor{reusepurple}{RGB}{105,0,184}
\definecolor{invoketeal}{RGB}{0,132,119}
\definecolor{paleblue}{RGB}{235,247,255}
\definecolor{palepurple}{RGB}{247,240,255}
\definecolor{paleteal}{RGB}{236,252,249}
\definecolor{focusorange}{RGB}{201,109,0}
\definecolor{paleorange}{RGB}{255,246,232}

\definecolor{lifeblue}{RGB}{46,102,178}
\definecolor{lifebluefill}{RGB}{241,246,252}
\definecolor{lifeviolet}{RGB}{120,71,188}
\definecolor{lifevioletfill}{RGB}{246,242,251}
\definecolor{lifegreen}{RGB}{39,137,97}
\definecolor{lifegreenfill}{RGB}{239,249,244}
\definecolor{lifeorange}{RGB}{201,118,35}
\definecolor{lifeorangefill}{RGB}{252,247,240}
\definecolor{lifesoftgray}{RGB}{75,75,75}
\definecolor{lifemidgray}{RGB}{145,145,145}
\definecolor{lifelightgrayfill}{RGB}{248,248,248}

\title{Recursive Formation and First-Class Reuse of Compositional Capabilities}
\author{Hang Lin}
\date{October 2026}

\begin{document}

\maketitle

\begin{abstract}
We introduce \FormCraft, a framework for representing acquired operations as first-class capabilities that can be recursively composed, directly invoked, and selectively retained in a persistent library through domain-specific verification. The framework separates capability well-formedness and public callability from later admission to persistent storage while preserving formation provenance and operational identity. We formalize recursive compositional closure over a uniform capability space and show that formed composites can be reused through the same public interface without re-forming their historical derivation. Experiments across continuous, discrete, propositional, and sequence domains demonstrate instance-independent reuse, independent checking, persistent acquisition, and reduced reuse-time search work under first-class exposure.
\end{abstract}

\section{Introduction}
\label{sec:introduction}

Program synthesis turns specifications or examples into executable programs, using search, constraints, and learned guidance \cite{gulwani2017synthesis,alur2013sygus,gulwani2011flashfill,balog2017deepcoder,devlin2017robustfill}. Library-learning systems go further by extracting abstractions from solved tasks and making them useful for later synthesis \cite{ellis2021dreamcoder,bowers2023stitch,grand2024lilo}. These advances motivate a question about the resulting artifacts: once a useful operation has been formed, what makes it a reusable operational object rather than only a derivation or a search result, and how can that object later be retained?

We study that question as a capability lifecycle. A formed operation should be represented independently of the instance that elicited it, composed through compatible public contracts, and exposed for later invocation as an operation in its own right. Long-term persistence is a separate concern: a capability may subsequently be checked before it is retained in a persistent library. The framework therefore centers on representation, composition, persistence, and reuse after formation, while abstraction-discovery methods can serve as upstream sources of capabilities. In elementary arithmetic, an operation is more reusable than the example that first demonstrated it. The analogy captures the intended separation between a concrete instance and an operation that accepts new operands.

\FormCraft{} represents primitive capabilities $f$ and composite capabilities $F$ within one well-formed capability space $\mathfrak F$. Instance data $X$ and parameters $\alpha$ enter through public invocation ports, while a composite retains its formation provenance. Compatible recursive composition produces another first-class capability in the same space; verification is invoked only when a formed capability is considered for persistent library membership. Proposition~\ref{prop:compositional-closure} establishes finite compositional closure, and Corollary~\ref{cor:first-class-depth} isolates the consequence of first-class exposure for historical re-formation depth.

The evaluation progresses from instance-independent reuse to independently checked composition, persistent acquisition, and increasing-depth first-class reuse. Continuous and discrete prototypes examine instance separation; a Lean-checked propositional microdomain tests independently verifiable composites \cite{demoura2021lean4}. In the main sequence-domain study, sequential acquisition expands the persistent store from 12 to 26 and then 38 records. A paired reuse experiment over the increasing-depth task family reduces candidate attempts from 16,570 to 144 when persisted composites are exposed directly through their public interfaces. The paired study reports reuse-time search work under controlled snapshots.

The contributions are:
\begin{enumerate}
  \item A formal account of structural axioms and first-class compositional closure for instance-independent capabilities.
  \item A lifecycle that separates well-formed capability formation and first-class reuse from the later verification decision governing persistent library membership.
  \item A record-based implementation that gives primitive and composite capabilities one public operational interface while retaining operation identity and formation provenance.
  \item Cross-domain evaluation of instance-independent reuse and independently checked composition, together with sequential persistent-store growth and a paired comparison of first-class reuse with historical expansion.
\end{enumerate}

Section~\ref{sec:formulation} gives the formalization; Section~\ref{sec:framework} describes its implementation; Section~\ref{sec:evaluation} reports the evidence; and Section~\ref{sec:related} positions the work.

\section{Formal Model}
\label{sec:formulation}

\subsection{Primitive Capability and Operational Semantics}

Let $X\in\mathcal X$ be a concrete task instance and $\alpha\in\mathcal A$ an external invocation parameter. A primitive capability is denoted by $f$, a composite by $F$, and a generic capability of either kind by $G$. The space $\mathfrak F$ contains well-formed capability representations. Its members have declared input, parameter, and output contracts and share the execution form
\begin{equation}
\label{eq:operational-semantics}
  \llbracket G\rrbracket(X;\alpha)=Y,\qquad G\in\mathfrak F.
\end{equation}
A representation records an operation and its contract, while a particular invocation supplies its own $X$ and $\alpha$. Formation provenance is denoted by $P(G)$. Throughout, notation such as $X\notin f$ is representational shorthand: it means that the capability representation of $f$ contains no concrete invocation payload $X$, not that $f$ is literally a set. Behavioral equivalence $\sim_F$ means equality of observable denotations under a shared compatible contract; it is defined below.

\subsection{Foundational Axioms}
\label{sec:structural-axioms}

\begin{axiom}[Primitive representation and recursive formation]
\label{ax:structural-rules}
\begin{align*}
  \mathrm{A}_1\quad \textnormal{(Instance Exclusion)}
    &:\quad X\notin f,\\
  \mathrm{A}_2\quad \textnormal{(External Parameterization)}
    &:\quad \alpha\ \text{is external},\\
  \mathrm{A}_3\quad \textnormal{(Uniform Operational Interface)}
    &:\quad G:\mathcal X_G\times\mathcal A_G\longrightarrow\mathcal Y_G
      \quad(G=f\text{ or }F),\\
  \mathrm{A}_4\quad \textnormal{(Recursive Formation)}
    &:\quad G_i,G_j\in\mathfrak F,\ 
      \operatorname{Compatible}(G_i,G_j,\omega),\\[-1pt]
    &\phantom{:\quad}\Longrightarrow
      F=\operatorname{Compose}_{\omega}(G_i,G_j)\in\mathfrak F.
\end{align*}
\end{axiom}

Axiom $\mathrm{A}_1$ excludes a concrete instance payload from the primitive representation without forbidding learning from instances. Axiom $\mathrm{A}_2$ permits a stored parameter schema but places invocation-specific values at public ports. By $\mathrm{A}_3$, a primitive and composite differ in formation history, not public operational status. The terminology follows the programming-language tradition of treating procedures and functions as first-class values \cite{strachey2000fundamental}; within the capability space, we use \emph{first-class} to mean that a formed composite is exposed through the same public interface as other capabilities: it can be directly invoked, retained, and supplied as an operand to later compatible formation.

In $\mathrm{A}_4$, compatibility covers contracts, ports, and the static rule $\omega$. The constructor records capability references, contracts, wiring, and provenance, but captures neither runtime $X$ nor $\alpha$. Either parent may already be composite, allowing $f\to F_1\to F_2\to\cdots$. Finite closure of this recursive rule is established separately by Proposition~\ref{prop:compositional-closure}.

\subsection{Recursive Composition and Semantic Associativity}

\begin{definition}[Compositional denotation]
\label{def:compositional-denotation}
For a well-formed composite $F=\operatorname{Compose}_{\omega}(G_i,G_j)$, its denotation is induced by the denotations of its parents and the compatible static wiring $\omega$:
\begin{equation}
\label{eq:composition-semantics}
  \llbracket F\rrbracket
  =\llbracket G_i\rrbracket\circ_{\omega}\llbracket G_j\rrbracket.
\end{equation}
Here $\circ_{\omega}$ denotes the behavior induced by the typed wiring rule rather than unconditional ordinary function composition.
\end{definition}

For any composite capability expression $E\in\mathfrak F^{*}$, let $W(E)$ denote its \emph{flattened typed wiring}: the public-port dependency structure obtained by recursively expanding composite nodes while preserving their external contracts. This flattening removes only the internal compositional nesting from the wiring view; formation provenance $P(E)$, parent references, and the recorded formation DAG remain recorded separately. We write $\circ$ below for the appropriate compatible $\operatorname{Compose}_{\omega}$ at each node.

\begin{proposition}[Semantic associativity under wiring equivalence]
\label{prop:semantic-associativity}
Suppose both $(G_1\circ G_2)\circ G_3$ and $G_1\circ(G_2\circ G_3)$ are contract-compatible and have the same flattened typed wiring,
\[
  W\bigl((G_1\circ G_2)\circ G_3\bigr)
  =
  W\bigl(G_1\circ(G_2\circ G_3)\bigr).
\]
Then
\[
  (G_1\circ G_2)\circ G_3\sim_F
  G_1\circ(G_2\circ G_3).
\]
Their formation provenance and representation identities may nevertheless differ.
\end{proposition}
\begin{proof}
By Definition~\ref{def:compositional-denotation}, recursively expanding either well-formed parenthesization evaluates the same parent denotations along the same flattened typed wiring. Their root denotations are therefore equal on the shared compatible contract, which is precisely behavioral equivalence under $\sim_F$. The recorded formation DAGs need not coincide.
\end{proof}
This associativity statement is semantic and conditional on wiring equivalence. It is independent of the finite-closure result below: recursive compatible formation yields closure by induction, while semantic associativity identifies compatible regroupings that realize the same flattened operation.

\subsection{Finite Compositional Closure}
\label{sec:compositional-closure}

Let $\mathfrak F_0\subseteq\mathfrak F$ be the primitives. Define
\begin{equation}
\label{eq:formation-levels}
  \mathfrak F_{n+1}
  =\mathfrak F_n\cup
  \left\{
    \operatorname{Compose}_{\omega}(G_i,G_j)
    \;\middle|\;
    \begin{gathered}
      G_i,G_j\in\mathfrak F_n,\\
      \operatorname{Compatible}(G_i,G_j,\omega)
    \end{gathered}
  \right\}.
\end{equation}
and $\mathfrak F^{*}=\bigcup_{n\geq0}\mathfrak F_n$.

\begin{proposition}[Finite compositional closure]
\label{prop:compositional-closure}
Every composite $F$ obtained from primitives by finite compatible formation belongs to $\mathfrak F^{*}\subseteq\mathfrak F$ and may serve as a parent in another compatible formation.
\end{proposition}
\begin{proof}
Primitive leaves lie in $\mathfrak F_0$. If two parents have finite depth, they lie in a common $\mathfrak F_n$; Axiom~\ref{ax:structural-rules} places their compatible one-step composition in $\mathfrak F_{n+1}$. Induction over a finite formation DAG proves the root claim, after which the same recursive formation rule permits another compatible composition.
\end{proof}
Proposition~\ref{prop:semantic-associativity} supplies a separate semantic fact: when two compatible regroupings have the same flattened typed wiring, they lie in the same behavioral equivalence class even though their formation DAGs may differ.

\subsection{First-Class Consequences}

\begin{corollary}[Preserved instance exclusion]
\label{cor:instance-exclusion}
For every finite composite $F$, $X\notin F$ in its representation.
\end{corollary}
The primitive base satisfies $\mathrm{A}_1$, and each constructor stores only static formation information.

\begin{corollary}[Preserved external parameterization]
\label{cor:external-parameterization}
A finite composite receives invocation-specific $\alpha$ externally:
$\llbracket F\rrbracket(X;\alpha)=Y$ for a contract-valid call.
\end{corollary}
Formation retains parameter contracts without capturing a particular call's value.

\begin{corollary}[Direct first-class invocation]
\label{cor:direct-invocation}
Each $F\in\mathfrak F^{*}$ has the public form of $\mathrm{A}_3$ and can accept a contract-valid new instance directly:
\begin{equation}
\label{eq:direct-invocation}
  (X,\alpha)\longrightarrow F\longrightarrow Y.
\end{equation}
\end{corollary}
Formation history is retained, but it is not required for public invocation. The executable body may still evaluate child operations internally.

\begin{corollary}[Recursive reusability]
\label{cor:recursive-reuse}
If a formed composite $F\in\mathfrak F^{*}$ and another capability $G\in\mathfrak F$ are contract-compatible under a compatible wiring rule $\omega$, then
\[
  \operatorname{Compose}_{\omega}(F,G)\in\mathfrak F.
\]
Either parent may itself be composite; formation creates no lower-status composite class.
\end{corollary}

Figure~\ref{fig:first-class-reuse} makes the distinction between formation history and later operational reuse explicit. The left panel shows one compatible recursive history in which $G_5$ is itself a formed capability and formation later continues to $G_6$. The right panel reuses that same $G_5$ as a single public capability. The orange styling only highlights the capability being followed across the two panels; it does not denote a different operational type. The historical subgraph and its provenance remain recorded, but they need not re-enter the public search view as formation decisions during later reuse.

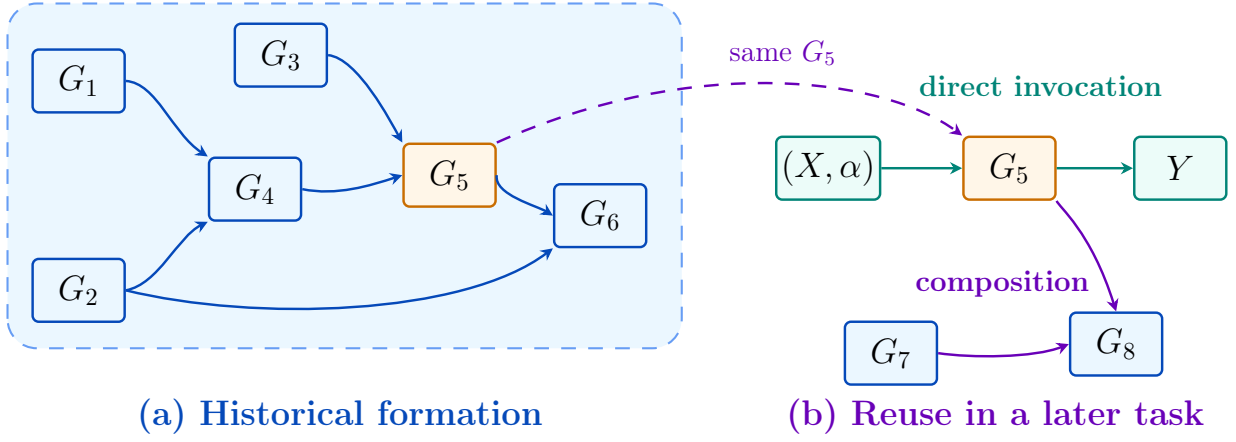
\begin{figure}[t]
\centering
\resizebox{\linewidth}{!}{%
\begin{tikzpicture}[
  >=Stealth,
  formationnode/.style={
    draw=formationblue, fill=paleblue,
    rounded corners=2.5pt, line width=1pt,
    minimum width=1.32cm, minimum height=0.90cm,
    inner sep=2pt, font=\Large
  },
  focusnode/.style={
    draw=focusorange, fill=paleorange,
    rounded corners=2.5pt, line width=1pt,
    minimum width=1.32cm, minimum height=0.90cm,
    inner sep=2pt, font=\Large
  },
  invocationnode/.style={
    draw=invoketeal, fill=paleteal,
    rounded corners=2.5pt, line width=1pt,
    minimum width=1.32cm, minimum height=0.90cm,
    inner sep=2pt, font=\Large
  },
  formationedge/.style={
    -{Stealth[length=5pt,width=5pt]},
    draw=formationblue, line width=1pt
  },
  reuseedge/.style={
    -{Stealth[length=5pt,width=5pt]},
    draw=reusepurple, line width=1pt
  },
  invocationedge/.style={
    -{Stealth[length=5pt,width=5pt]},
    draw=invoketeal, line width=1pt
  }
]

\path[fill=white] (0,0) rectangle (17.9,8.9);

\node[font=\LARGE] at (8.95,8.18)
  {$\operatorname{Compose}_{\omega}:
    \mathfrak F\times\mathfrak F
    \rightharpoonup\mathfrak F$};

\path[fill=paleblue,rounded corners=12pt]
  (0.28,1.72) rectangle (9.95,6.66);

\draw[
  draw=borderblue, line width=0.85pt,
  dash pattern=on 6pt off 5pt,
  rounded corners=12pt
] (0.28,1.72) rectangle (9.95,6.66);

\node[formationnode] (G1) at (1.30,5.55) {$G_1$};
\node[formationnode] (G2) at (1.30,2.55) {$G_2$};
\node[formationnode] (G3) at (4.20,5.92) {$G_3$};
\node[formationnode] (G4) at (3.83,4.00) {$G_4$};

\node[focusnode] (G5old) at (6.62,4.20) {$G_5$};
\node[formationnode] (G6) at (8.78,3.62) {$G_6$};

\draw[formationedge]
  (G1.east) .. controls (2.42,5.52) and (2.64,4.80) ..
  (G4.north west);

\draw[formationedge]
  (G2.east) .. controls (2.45,2.64) and (2.72,3.32) ..
  (G4.south west);

\draw[formationedge]
  (G3.east) .. controls (5.30,5.88) and (5.78,4.96) ..
  (G5old.north west);

\draw[formationedge]
  (G4.east) .. controls (5.00,3.90) and (5.48,4.05) ..
  (G5old.west);

\draw[formationedge]
  (G5old.east) .. controls (7.25,3.92) and (7.88,3.76) ..
  (G6.west);

\draw[formationedge]
  (G2.east) .. controls (3.82,2.10) and (6.82,2.18) ..
  (G6.south west);

\node[focusnode] (G5reuse) at (14.65,4.30) {$G_5$};

\draw[
  -{Stealth[length=5pt,width=5pt]},
  draw=reusepurple, line width=1pt,
  dash pattern=on 6pt off 5pt
]
  (G5old.north east) .. controls (10.15,5.95) and (12.85,5.60) ..
  (G5reuse.north west);

\node[text=reusepurple,font=\large]
  at (11.40,5.95) {same $G_5$};

\node[text=invoketeal,font=\large\bfseries]
  at (15.10,5.48) {direct invocation};

\node[invocationnode] (input) at (12.05,4.30) {$(X,\alpha)$};
\node[invocationnode] (output) at (17.10,4.30) {$Y$};

\draw[invocationedge] (input.east) -- (G5reuse.west);
\draw[invocationedge] (G5reuse.east) -- (output.west);

\node[text=reusepurple,font=\large\bfseries]
  at (14.55,2.65) {composition};

\node[formationnode] (G7) at (12.95,1.65) {$G_7$};
\node[formationnode] (G8) at (16.18,1.78) {$G_8$};

\draw[reuseedge]
  (G5reuse.south east)
  .. controls (15.35,3.75) and (15.82,3.40) ..
  (G8.north);

\draw[reuseedge]
  (G7.east) .. controls (14.10,1.58) and (15.00,1.58) ..
  (G8.west);

\node[text=formationblue,font=\Large\bfseries]
  at (5.05,0.74) {(a) Historical formation};

\node[text=reusepurple,font=\Large\bfseries]
  at (14.45,0.74) {(b) Reuse in a later task};

\end{tikzpicture}%
}
\caption{Recursive formation and first-class reuse of an intermediate capability. All boxes denote generic capabilities $G\in\mathfrak F$; the orange styling only highlights the same $G_5$ across the two panels and does not denote a distinct operational type. In (a), compatible recursive composition forms $G_5$ and later continues to $G_6$. In (b), the same $G_5$ is exposed directly in a later task: it can accept a new invocation $(X,\alpha)$ or participate in further compatible composition with $G_7$ to form $G_8$, without re-forming its historical subgraph. The recorded formation DAG and provenance remain retained.}
\label{fig:first-class-reuse}
\end{figure}

\begin{corollary}[Historical re-formation depth under first-class reuse]
\label{cor:first-class-depth}
If $F$ is available to search through its public interface, let $h_{\mathrm{form}}(F)$ denote its historical formation depth, $h_{\mathrm{public}}(F)$ its public search representation depth, and $h_{\mathrm{reform}}(F)$ the new depth required solely to reconstruct its historical chain. For any finite $h_{\mathrm{form}}(F)$,
\begin{equation}
\label{eq:reformation-depth}
  h_{\mathrm{public}}(F)=1,\qquad h_{\mathrm{reform}}(F)=0.
\end{equation}
\end{corollary}
Search sees one public operation and need not form its ancestors again. The corollary concerns the public search representation of the capability; execution continues to follow the stored executable body, and selection remains part of the surrounding search procedure.

\subsection{Reuse, Identity, and Quotient Spaces}

The provenance rule is $\mathrm{R}_1:\ P(\operatorname{Reuse}(F))=P(F)$. Repeated application gives, for finite $n\in\mathbb N_0$,
\begin{equation}
\label{eq:provenance-preservation}
  P(\operatorname{Reuse}^{n}(F))=P(F).
\end{equation}
Reuse preserves formation identity: new outputs or execution records leave the recorded formation provenance unchanged.

Within a fixed observable contract, for $G_a,G_b\in\mathfrak F$, behavioral equivalence is defined by
\begin{equation}
\label{eq:behavioral-equivalence}
  G_a\sim_F G_b
  \quad\Longleftrightarrow\quad
  \llbracket G_a\rrbracket(X;\alpha)
  =\llbracket G_b\rrbracket(X;\alpha)
  \quad\text{for all contract-valid }(X,\alpha).
\end{equation}
The identity rule $\mathrm{R}_2$ permits, but does not require, $\operatorname{opid}(G_a)=\operatorname{opid}(G_b)$ when $G_a\sim_F G_b$. Distinct provenance $P(G_a)\ne P(G_b)$ may remain. The regroupings of Proposition~\ref{prop:semantic-associativity} illustrate equivalent operations with different formation DAGs.

\begin{corollary}[Operational and formation identities are separable]
\label{cor:identity-separation}
A shared operation identity for behaviorally equivalent capabilities does not require equal formation provenance.
\end{corollary}

Quotienting the well-formed capability space by semantic equivalence gives
\begin{equation}
\label{eq:capability-quotient}
  \mathcal C=\mathfrak F/\!\sim_F.
\end{equation}
A class $[G]_{\sim_F}$ may contain different representation identities, formation DAGs, and provenance while representing the same observable operation. The finitely generated compositional region defined above induces the subset
\[
  \mathcal C^{*}
  :=\bigl\{[G]_{\sim_F}\in\mathcal C\mid G\in\mathfrak F^{*}\bigr\}
  \subseteq\mathcal C.
\]
If behavioral equivalence is a congruence for contract-compatible composition, $G_a\sim_F G_a'$ and $G_b\sim_F G_b'$ imply $G_a\circ G_b\sim_F G_a'\circ G_b'$. Composition then descends to compatible equivalence classes $[G_a]\circ[G_b]$; it remains partial because some contracts cannot be composed.

Let $\mathcal G$ be tasks with requested observable behavior. Under a compatible observable contract, $g_1\sim_G g_2$ means $\llbracket g_1\rrbracket=\llbracket g_2\rrbracket$, yielding
\begin{equation}
\label{eq:task-quotient}
  \mathcal T=\mathcal G/\!\sim_G.
\end{equation}
The capability quotient $\mathcal C$ and task quotient $\mathcal T$ play distinct semantic roles, and no canonical identification between them is assumed: one capability class may serve several task classes, while one task class may admit multiple operational realizations. Semantic equivalence is an ideal definition; finite experiments approximate it with domain-specific evidence. The next section describes how these capability objects and formation relations are represented and managed in the implementation.

\section{The \FormCraft{} Framework}
\label{sec:framework}

\subsection{Executable Capability Records}

A well-formed capability $G\in\mathfrak F$ is represented in the implementation by an executable \emph{capability record}. In the retained implementation, this record is realized by the class \texttt{Structure}. An executable capability record contains a record identity, public operation identity, typed input and output ports, an external-parameter schema, and an executable body. A composite record additionally stores parent identities, a static wiring rule, formation-level metadata, and formation provenance. These fields realize Equation~\eqref{eq:operational-semantics}. The persistent store may also contain auxiliary non-executable records, such as topology templates; persistent-record counts and executable-capability counts are therefore tracked separately. Storage location is independent of a capability's well-formedness and public call form.

The record identity distinguishes serialized objects, $\operatorname{opid}(F)$ denotes operational identity, and $P(F)$ records a particular formation history. Equivalent operations may retain different parent graphs. Serialization and persistent hashes preserve those distinctions when a capability record is retained for later use.

Table~\ref{tab:structure-schema} summarizes the implementation fields that connect the formal objects to the executable record. Primitive and composite capabilities share the same public record shape; composite-specific fields add formation relations without changing the operational interface.

\begin{table}[t]
\centering
\small
\setlength{\tabcolsep}{7pt}
\renewcommand{\arraystretch}{1.12}
\begin{tabular}{@{}>{\raggedright\arraybackslash}p{0.29\linewidth}>{\raggedright\arraybackslash}p{0.63\linewidth}@{}}
\toprule
\textbf{Capability-record field} & \textbf{Role in the lifecycle} \\
\midrule
Record identity & Distinguishes concrete records and serialized objects. \\
\addlinespace[5pt]
$\operatorname{opid}(G)$ & Represents operational identity independently of a particular formation history. \\
\addlinespace[5pt]
Typed input/output ports & Declare the public contract used for invocation and compatibility checks. \\
\addlinespace[5pt]
External-parameter schema & Declares the form of invocation-specific $\alpha$ while keeping call-specific values external. \\
\addlinespace[5pt]
Executable body & Implements the observable operation $\llbracket G\rrbracket$. \\
\addlinespace[5pt]
Parent references and wiring $\omega$ & Record the static formation relation for a composite capability record. \\
\addlinespace[5pt]
Formation level & Records formation-depth metadata (stored internally as \texttt{canonical\_level}). \\
\addlinespace[5pt]
Formation provenance $P(G)$ & Retains the formation history used for audit and later inspection. \\
\bottomrule
\end{tabular}
\caption{Core fields of an executable capability record. Composite-specific formation fields extend the record while preserving the same public call form.}
\label{tab:structure-schema}
\end{table}

\subsection{Capability Formation and Composition}

Primitive acquisition, explicit authorship, and automatic formation can produce executable capability records with different provenance. Composition consumes public interfaces, checks contracts and ports, and records exact parent references and the rule $\omega$. Recursive formation implements Axiom~\ref{ax:structural-rules}; finite closure follows from Proposition~\ref{prop:compositional-closure}. The resulting $F$ can be invoked and composed again before any long-term persistence decision.

The implementation therefore applies three operations to the same capability object: formation produces a well-formed $F\in\mathfrak F$; admission may retain that capability in the persistent library; later reuse exposes the same public interface for invocation or further composition. In compact form,
\begin{equation*}
{\renewcommand{\arraystretch}{1.35}
\begin{array}{r@{\qquad}r@{\quad}c@{\quad}l}
  \text{Formation:} &
    (G_i,G_j,\omega) &
    \makebox[7.5em][c]{\ensuremath{\longmapsto}} &
    F\in\mathfrak F,\\
  \text{Admission:} &
    F\in\mathfrak F &
    \makebox[7.5em][c]{\ensuremath{\xrightarrow{\text{verification}}}} &
    F\in\Library,\\
  \text{Reuse:} &
    (X,\alpha,F) &
    \makebox[7.5em][c]{\ensuremath{\longmapsto}} &
    Y\quad\text{or}\quad(F,G')\longmapsto F'.
\end{array}
}
\end{equation*}
Formation, optional persistence, direct invocation, and further composition therefore refer to one operational capability whose provenance and parent references remain attached throughout the lifecycle.

In the sequence-domain implementation, public operations use the contract $\texttt{sequence<int>}\to\texttt{int}$, and typed binary composition combines existing statistics. This concrete contract instantiates the general first-class interface defined by the formal model.

\subsection{Verification and Library Admission}

Long-term persistence is a separate lifecycle decision. We reserve $\Library\subseteq\mathfrak F$ for admitted executable capabilities; auxiliary non-executable records may coexist in the broader persistent store but are not members of $\Library$. A well-formed first-class capability can be checked against its declared contract and admitted into $\Library$:
\begin{equation}
\label{eq:admission}
  F\in\mathfrak F
  \xrightarrow{\text{verification}}
  F\in\Library,
  \qquad
  \Library\subseteq\mathfrak F.
\end{equation}
The formation rules determine membership in $\mathfrak F$ and first-class operational status; verification governs the later persistence decision. In the retained implementation, candidates may be staged in an external KV store before this gate. Admission adds the capability to $\Library$, while library membership remains distinct from well-formedness in $\mathfrak F$.

Depending on the domain, checks may include public interface and contract validation, executable behavior, examples, randomized oracle cases, cross-language agreement, cold-instance execution, instance-leakage and exact-parent audits, and persistent reload. Receipts bind promotion to checks and a source digest. Lean provides the independent proof checker used by the propositional microdomain \cite{demoura2021lean4}; the other domains instantiate verification through their executable contracts and oracles.

\subsection{First-Class Capability Reuse}

A capability has its first-class call form before persistence. Library membership makes the corresponding admitted capability record available for long-term retrieval through that same interface. Historical expansion reconstructs ancestor operations as search decisions. First-class exposure selects $F$ as one public operation and binds its slots directly, as in Equation~\eqref{eq:direct-invocation}.

Parent references and provenance remain available for audit, and execution may evaluate constituent operations internally. The search-level distinction is whether those historical formation decisions must be reconstructed before the stored capability can be selected. When $F$ is exposed through its public interface, Corollary~\ref{cor:first-class-depth} applies; Section~\ref{sec:evaluation} measures that effect in a paired increasing-depth comparison.

\section{Experimental Evaluation}
\label{sec:evaluation}

The evidence follows the capability lifecycle. Single-operation prototypes examine instance-independent reuse across continuous and discrete domains; the Lean microdomain examines independently checkable composition; the main sequence-domain study then evaluates persistent acquisition, increasing formation depth, and the reuse-time effect of exposing persisted composites directly. Each experiment family targets a different transition in the lifecycle, and quantitative counts are reported within their corresponding snapshots.

Figure~\ref{fig:capability-lifecycle} summarizes the experimental lifecycle and separates three questions: how a well-formed first-class capability is formed, whether it is admitted for persistent retention, and how a persisted capability is later exposed for reuse. The admission stage therefore governs membership in $\Library$, while well-formedness in $\mathfrak F$ and first-class status are determined earlier by the representation and formation rules.

\begin{figure}[H]
\centering
\resizebox{\linewidth}{!}{%
\begin{tikzpicture}[
  >=Stealth,
  lifemainbox/.style={
    draw,fill=white,rounded corners=4pt,line width=1pt,
    minimum height=1.62cm,inner sep=5pt,
    align=center,font=\large
  },
  lifesmallnote/.style={
    font=\footnotesize,text=lifesoftgray,align=center
  },
  lifenotetag/.style={
    font=\footnotesize,text=lifesoftgray,fill=white,
    inner sep=1.5pt,align=center
  },
  liferejectedbox/.style={
    draw=lifemidgray,fill=lifelightgrayfill,
    rounded corners=4pt,line width=0.9pt,
    minimum height=1.35cm,inner sep=5pt,
    align=center,font=\normalsize
  },
  lifeformedge/.style={
    -{Stealth[length=5.5pt,width=5.5pt]},
    draw=lifeblue,line width=1pt
  },
  lifeverifyedge/.style={
    -{Stealth[length=5.5pt,width=5.5pt]},
    draw=lifeviolet,line width=1pt
  },
  lifelibraryedge/.style={
    -{Stealth[length=5.5pt,width=5.5pt]},
    draw=lifegreen,line width=1pt
  },
  lifereuseedge/.style={
    -{Stealth[length=5.5pt,width=5.5pt]},
    draw=lifeorange,line width=1pt
  },
  liferejectededge/.style={
    -{Stealth[length=5pt,width=5pt]},
    draw=lifemidgray,line width=0.9pt,
    dash pattern=on 5pt off 4pt
  }
]

\path[fill=white] (0,0) rectangle (32.5,10.7);

\node[font=\LARGE\bfseries]
  at (16.25,10.00)
  {Formation, admission, and first-class reuse};

\path[fill=lifebluefill,rounded corners=11pt]
  (0.35,1.35) rectangle (10.10,8.90);

\path[fill=lifevioletfill,rounded corners=11pt]
  (10.30,1.35) rectangle (22.65,8.90);

\path[fill=lifeorangefill,rounded corners=11pt]
  (22.85,1.35) rectangle (32.05,8.90);


\node[
  circle,fill=lifeblue,text=white,
  minimum size=9.5mm,inner sep=0pt,font=\bfseries\Large
] at (1.15,8.30) {1};

\node[
  text=lifeblue,font=\Large\bfseries,anchor=west
] at (1.85,8.30)
  {Formation};

\draw[lifeblue,line width=0.9pt]
  (4.10,8.30) -- (9.40,8.30);

\node[
  circle,fill=lifeviolet,text=white,
  minimum size=9.5mm,inner sep=0pt,font=\bfseries\Large
] at (11.10,8.30) {2};

\node[
  text=lifeviolet,font=\Large\bfseries,anchor=west
] at (11.80,8.30)
  {Verification and admission};

\draw[lifeviolet,line width=0.9pt]
  (18.35,8.30) -- (21.95,8.30);

\node[
  circle,fill=lifeorange,text=white,
  minimum size=9.5mm,inner sep=0pt,font=\bfseries\Large
] at (23.60,8.30) {3};

\node[
  text=lifeorange,font=\Large\bfseries,anchor=west
] at (24.30,8.30)
  {Later first-class reuse};

\draw[lifeorange,line width=0.9pt]
  (29.75,8.30) -- (31.45,8.30);


\node[
  lifemainbox,draw=lifeblue,text width=2.25cm
] (LifeGoal) at (1.85,5.95)
{
  \textbf{Goal /}\\
  \textbf{unseen task}
};

\node[
  lifemainbox,draw=lifeblue,text width=2.35cm
] (LifeSearch) at (5.00,5.95)
{
  \textbf{Search or}\\
  \textbf{acquisition}
};

\node[
  lifemainbox,draw=lifeblue,text width=2.75cm
] (LifeLegal) at (8.35,5.95)
{
  \textbf{Well-formed capability}\\[3pt]
  $G\in\mathfrak F$
};

\draw[lifeformedge] (LifeGoal.east) -- (LifeSearch.west);
\draw[lifeformedge] (LifeSearch.east) -- (LifeLegal.west);

\node[lifesmallnote,text width=2.35cm]
  at (1.85,4.05)
{
  task instance\\
  or specification
};

\node[lifesmallnote,text width=2.80cm]
  at (5.00,3.95)
{
  enumeration, retrieval,\\
  or composition
};

\node[lifesmallnote,text width=3.00cm]
  at (8.35,4.05)
{
  first-class capability\\
  before persistence
};


\node[
  lifemainbox,
  draw=lifeviolet,
  text width=3.95cm,
  minimum height=2.30cm
] (LifeVerify) at (12.75,5.95)
{
  {\bfseries\large Verification suite}\\[6pt]
  {\normalsize contract / interface}\\[4pt]
  {\normalsize holdout behavior}\\[4pt]
  {\normalsize reload / persistence}
};

\node[
  lifemainbox,
  draw=lifeviolet,
  text width=2.35cm
] (LifeAdmission) at (16.95,5.95)
{
  \textbf{Admission}\\
  \textbf{decision}
};

\node[
  lifemainbox,
  draw=lifegreen,
  fill=lifegreenfill,
  text width=3.10cm,
  minimum height=1.75cm
] (LifePersist) at (20.95,5.95)
{
  \textbf{Persistent member}\\[3pt]
  $G\in\mathcal L$
};

\draw[lifeverifyedge] (LifeLegal.east) -- (LifeVerify.west);
\draw[lifeverifyedge] (LifeVerify.east) -- (LifeAdmission.west);

\draw[lifelibraryedge]
  (LifeAdmission.east)
  -- node[
       above,
       text=lifegreen,
       font=\footnotesize,
       fill=white,
       inner sep=1pt
     ] {admit}
  (LifePersist.west);

\node[lifesmallnote,text width=3.40cm]
  at (20.95,3.95)
{
  persistent membership\\
  is a library-level property
};

\node[
  liferejectedbox,
  text width=4.20cm
] (LifeRejected) at (16.95,2.60)
{
  \textbf{not promoted to $\mathcal L$}\\[2pt]
  remains in $\mathfrak F$
};

\draw[liferejectededge]
  (LifeAdmission.south) -- (LifeRejected.north);

\node[lifenotetag,anchor=east]
  at (16.48,4.12)
  {not admitted};


\node[
  lifemainbox,
  draw=lifeorange,
  text width=3.55cm,
  minimum height=1.75cm
] (LifeExposure) at (27.70,5.95)
{
  \textbf{Later first-class}\\
  \textbf{exposure}
};

\draw[lifelibraryedge]
  (LifePersist.east) -- (LifeExposure.west);

\node[
  lifemainbox,
  draw=lifeorange,
  minimum width=4.10cm,
  minimum height=1.40cm,
  font=\large\bfseries
] (LifeDirect) at (25.20,2.98)
{
  Direct invocation
};

\node[
  lifemainbox,
  draw=lifeorange,
  minimum width=4.35cm,
  minimum height=1.40cm,
  font=\large\bfseries
] (LifeCompose) at (29.95,2.98)
{
  Further composition
};

\draw[
  draw=lifeorange,
  line width=1pt
]
  (LifeExposure.south) -- (27.70,4.38);

\draw[lifereuseedge]
  (27.70,4.38) -- (25.20,4.38) -- (LifeDirect.north);

\draw[lifereuseedge]
  (27.70,4.38) -- (29.95,4.38) -- (LifeCompose.north);

\node[font=\large]
  at (25.20,1.78)
{
  $(X,\alpha)\rightarrow G\rightarrow Y$
};

\node[font=\large]
  at (29.95,1.78)
{
  $G,G'\rightarrow G''$
};

\end{tikzpicture}%
}
\caption{Formation--admission--reuse lifecycle used to organize the experimental evidence. Search or acquisition first produces a well-formed first-class capability $G\in\mathfrak F$. Verification and admission determine persistent membership $G\in\Library$, while membership in $\mathfrak F$ remains a property of the capability representation. An admitted capability can later be retrieved through the same public interface for direct invocation or further composition. The applicable verification procedure is domain-dependent.}
\label{fig:capability-lifecycle}
\end{figure}
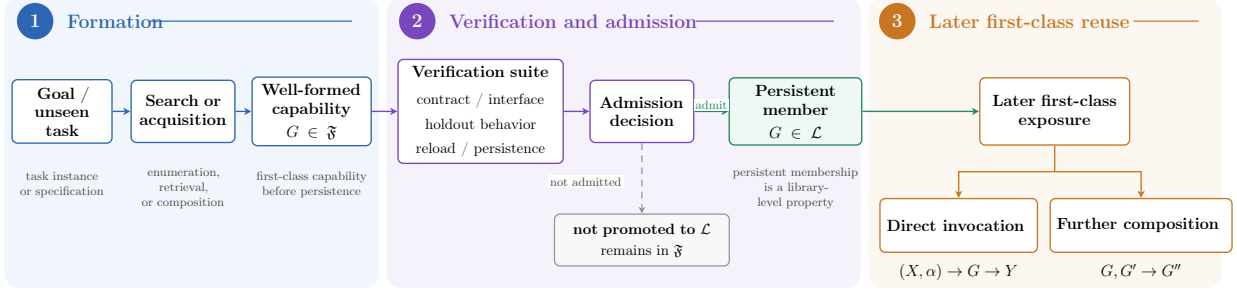

The evaluation covers four complementary settings. Continuous and character-substitution experiments examine instance-independent reuse across numerical and discrete domains; the Lean propositional microdomain examines independent checking after composition; the sequence-domain acquisition stages evaluate persistence and recursive library growth; and the paired D01--D12 comparison measures the reuse-time effect of exposing formed composites directly through their public interfaces.

\subsection{Experimental Setup}
\label{sec:experimental-setup}

The experimental families are early continuous-domain reuse, a discrete character-substitution prototype, a propositional composition microdomain checked by Lean, and two consecutive sequence-domain acquisition stages followed by a paired increasing-depth reuse comparison. Each family targets a distinct part of the formation--admission--reuse lifecycle. The continuous and character-substitution studies establish the instance/operation separation across two representation regimes; the Lean microdomain tests independent checking after composition; and the sequence-domain snapshots provide the main quantitative lifecycle and search-work measurements.

Stage I starts with 12 persistent records and seven example-only goals. It uses 336 fresh training instances and 192 root holdouts per goal, then checks 7/7 roots after cold reload. The resulting 26-record persistent store contains 24 executable capability records and two non-executable topology-template records. Stage II starts from those 26 persistent records with 12 fixed D01--D12 goals, no training examples, and 128 seeded holdout inputs per goal. All 12/12 goals are admitted, leaving 38 persistent records, of which 36 are executable capability records, with maximum formation level ten. Because the two acquisition stages use different protocols, their holdout counts are reported separately.

The paired comparison uses the same 12 D goals and starting snapshot in both views. Task definitions, applicable examples and holdouts, oracle, seed, verifier, and candidate limit are held fixed; only the historical search representation changes between \emph{Expanded} and \emph{First-class}. Candidate attempts count attempted search candidates; composition expansions count reconstruction of composed search nodes; intermediate materializations count intermediate objects created during that reconstruction; and verification batches count grouped checks. D-level derivation depth counts newly formed historical chain levels in the search route, not runtime call depth. The same definitions and controlled snapshot are used in both views.

\paragraph{Execution environment and reproducibility.}
The archived sequence-domain artifact contains the goal definitions, fixed seeds and digests, starting-library snapshot, per-view traces, and acquisition and admission summaries. Its package specification requires Python 3.10 or newer and NumPy 1.24 or newer. The sequence-domain implementation also checks results against a JavaScript executor. The reported quantitative results are search-work counts recorded in these archived experiment snapshots.

\paragraph{Code availability.}
The public implementation corresponding to the sequence-domain acquisition and paired increasing-depth reuse experiments is available at \url{https://github.com/lancelot138505-cpu/f-system}. The released implementation is contained in \path{04_depth_control_V9-4_complete_baseline/} and corresponds to the V9.4 implementation used for the main quantitative study.

\subsection{Instance-Independent Reuse Across Domains}

Early continuous-domain prototypes tested whether a single instance-independent operation could be reused across unseen inputs under externally supplied parameters. The discrete character-substitution prototype checked the same separation outside continuous numerical domains. Together, these prototypes provide qualitative evidence that the same instance-separation principle can be expressed in both continuous and discrete settings.

\subsection{Independent Checking of Composite Capabilities}

The propositional microdomain asks whether a composite remains independently checkable after formation. Five composite capabilities were formed from primitive proof operations. Each was reused on two theorem instances, including formulas not used to define its operation. Lean accepted all ten generated theorem checks \cite{demoura2021lean4}. Table~\ref{tab:lean} summarizes the result.

\begin{table}[H]
\centering
\small
\begin{tabular}{lr}
\toprule
Measure & Result \\
\midrule
Composite capabilities formed & 5/5 \\
Composite theorem reuses & 10/10 \\
Lean checks accepted (exit code 0) & 10/10 \\
Capability-record hash and formation provenance preserved & Yes \\
\bottomrule
\end{tabular}
\caption{Composite capability reuse in the Lean-checked propositional microdomain.}
\label{tab:lean}
\end{table}

The 5/5 formation and 10/10 theorem reuse results establish independent checkability of formed composites in the propositional microdomain. The next stage asks whether well-formed capabilities can also pass the persistence gate, survive reload, and remain available as parents for later formation.

\subsection{Sequential Acquisition and Persistent-Store Growth}

Stage I admits a step capability and a root capability for each of seven goals, adding 14 persistent executable capabilities to the 12-record seed store. Parent audits and semantic checks pass, and 7/7 roots pass cold reload. Stage II starts from the resulting 26 persistent records and admits 12 new D-level executable capabilities under the fixed D01--D12 task family. Its 12/12 admission and oracle checks pass; the maximum formation level rises from four to ten.

Table~\ref{tab:growth} summarizes the sequential growth from 12 to 26 and then 38 persistent records. The first stage tests example-driven acquisition and cold persistence; the second tests whether persisted executable capabilities can continue to support compatible composition while the broader record store grows. The different training protocols remain separate in the setup and interpretation.

\begin{table}[H]
\centering
\small
\begin{tabular}{p{0.49\linewidth}rrrr}
\toprule
Stage & Start & Goals & Added $F$ & End \\
\midrule
Stage I: Example-driven acquisition & 12 & 7 & 14 & 26 \\
Stage II: Increasing-depth chain (D01--D12) & 26 & 12 & 12 & 38 \\
\bottomrule
\end{tabular}
\caption{Sequential persistent-store growth in the sequence-domain study. Start and End count persistent records; Added $F$ counts new executable capabilities.}
\label{tab:growth}
\end{table}

\subsection{First-Class Reuse under Increasing Formation Depth}

The paired depth experiment compares the two search views over the same 12 D goals, instantiating Corollary~\ref{cor:first-class-depth}. Both views use the same frozen persistent snapshot and executable capability pool. \emph{Expanded} masks the target chain and re-forms its earlier D-level capabilities; \emph{First-class} exposes the already admitted target through its public interface. Trace audits confirm that first-class selection binds public slots without inspecting the target's parent graph before selection.

Table~\ref{tab:depth} reports reuse-time search-work counts. Candidate attempts fall from 16,570 to 144, a 99.1\% reduction; binding-combination counts equal candidate attempts in both views. Repeated historical composition expansions fall from 37,300 to zero, intermediate materializations from 16,778 to zero, and verification batches from 17,116 to 192. The route-level depth of newly formed D-level derivations changes from 78 to zero.
The 12 main-chain acquisitions pass admission, while the outputs of all paired search views match the frozen oracle.

\begin{table}[H]
\centering
\small
\begin{tabular}{lrrr}
\toprule
Search-work measure & Expanded & First-class & Reduction \\
\midrule
Candidate attempts & 16,570 & 144 & 99.1\% \\
Composition expansions & 37,300 & 0 & 100\% \\
Intermediate materializations & 16,778 & 0 & 100\% \\
Verification batches & 17,116 & 192 & 98.9\% \\
\bottomrule
\end{tabular}
\caption{Paired reuse-time search-work counts over 12 goals. Outputs in both search views matched the frozen oracle on evaluated cases; the 12 main-chain acquisitions passed admission.}
\label{tab:depth}
\end{table}

The paired counts isolate historical search reconstruction. The admitted $F$ retains its formation provenance and may evaluate constituent operations internally; Equation~\eqref{eq:reformation-depth} concerns the amount of historical formation that re-enters later search. In the same sequence-domain snapshots, the FULL and ACTIVE lifecycle views expose the same aggregate 366 executable records across the 12 goals and both yield 144 first-class candidate attempts. The observed reduction therefore follows from exposing formed composites directly rather than from a change in the available capability breadth.

\section{Related Work}
\label{sec:related}

\paragraph{Programming-language abstraction.}
Programming languages have long supported abstraction boundaries that separate public operations from concrete representation and have treated procedures or functions as values that can be manipulated through ordinary program structure \cite{strachey2000fundamental,liskov1974adt}. \FormCraft{} builds on this abstraction tradition at the level of acquired capability records: recursive formation produces another record with the same public operational status, while provenance and persistence decisions remain explicit parts of the lifecycle.

\paragraph{Program synthesis.}
Program synthesis searches for programs satisfying examples or other specifications; representative methods include programming by example, enumerative and constraint-based search, and learned guidance \cite{gulwani2017synthesis,gulwani2011flashfill,alur2013sygus}. Oracle-guided component synthesis couples candidate construction to semantic checks \cite{jha2010oracle}, while FlashMeta provides a framework for building inductive synthesizers over domain-specific languages \cite{polozov2015flashmeta}. DeepCoder uses learned predictions to guide search, and RobustFill studies neural synthesis from input-output examples \cite{balog2017deepcoder,devlin2017robustfill}. These works primarily address how programs are generated or selected. \FormCraft{} focuses on the operational status of a formed capability: how it is represented, composed, retained, and exposed again through a stable public interface.

\paragraph{Library and abstraction learning.}
DreamCoder alternates program search, library expansion, and recognition-guided synthesis \cite{ellis2021dreamcoder}. Stitch synthesizes abstractions from recurring structure in a corpus of programs \cite{bowers2023stitch}. LILO combines language-model-guided synthesis, symbolic compression, and documentation to learn interpretable libraries \cite{grand2024lilo}. Recent cumulative inductive-programming work provides a particularly close example: the Minimalist Machine treats learned predicates as a growing concept library, makes them available as background knowledge for later tasks, and uses that accumulated structure to discover recursive arithmetic programs \cite{magiera2026minimalist}. These systems provide strong examples of learned components becoming useful in later synthesis. \FormCraft{} approaches the same broad area from the representation side: recursive closure is defined over one uniform capability type, formation provenance remains attached to the resulting object, and persistent library admission is handled through a separate verifier gate. Abstractions produced by other methods can enter the same lifecycle whenever they satisfy the declared operational contracts and representation rules.

\paragraph{Independent checking.}
Formal proof assistants provide a distinct way to validate claims within a logical domain. Lean 4 is a theorem prover with a small trusted kernel and an extensible programming environment \cite{demoura2021lean4}. In the propositional microdomain, generated composite theorem reuses are submitted to Lean as an independent check. Across the framework, admission is parameterized by the declared contract and the checker available for that domain, allowing proof checking, executable oracles, cross-language agreement, and reload checks to occupy the same lifecycle role.

\section{Discussion and Conclusion}
\label{sec:discussion-conclusion}

The central design choice in \FormCraft{} is to make the result of compatible composition an operational object of the same public type as its parents. Once a composite has been formed, later search can encounter that capability through its public interface instead of reconstructing the sequence of formation decisions that originally produced it. The formation DAG remains recorded for audit and execution, but the public search view can treat the composite as one reusable operation. This is the operational content of Corollary~\ref{cor:first-class-depth}.

The distinction between formation depth and execution depth is important for interpreting this result. A stored composite can still evaluate constituent operations internally, and its executable body may therefore contain a nontrivial call structure. First-class exposure changes the representation presented to later search: the historical chain does not have to be generated again as a sequence of candidate-formation decisions before the composite can be selected. The paired increasing-depth experiment measures this search-level effect directly through candidate attempts, composition expansions, intermediate materializations, verification batches, and newly formed D-level derivation depth.

Persistence serves a different role. Membership in $\Library$ is a memory policy over already well-formed capabilities: it determines which formed operations remain available for later retrieval. This separation allows verification to govern long-term retention without redefining the capability's public operational status. The same distinction also keeps provenance meaningful, because reuse can preserve both operational identity and the original formation record rather than replacing history with a newly synthesized copy.

The experiment sequence reflects these layers. The continuous and character-substitution prototypes show the instance/operation separation in two data regimes. The Lean microdomain shows that a formed composite can remain independently checkable. The two sequence-domain acquisition stages then demonstrate persistent growth and recursive reuse, and the paired depth experiment isolates the effect of exposing already formed composites directly. Together, the experiments connect the formal lifecycle to executable records and to measured search behavior.

The present quantitative evaluation concentrates on the fixed D01--D12 task family and controlled persistent snapshots. Within that setting, the FULL and ACTIVE lifecycle views expose the same aggregate 366 executable records across the 12 goals and both yield 144 first-class candidate attempts, so the measured reduction is associated with historical-depth handling rather than a change in available breadth. Larger capability libraries, explicit breadth control, equivalence consolidation, and learned proposal or retrieval mechanisms are natural next variables for extending the same lifecycle. Stronger formal contracts and additional independently checked domains provide a parallel route for extending the verification side of the system.

In the controlled sequence-domain study, the persistent store grows from 12 to 26 and then 38 records. Exposing persisted composites through their first-class interfaces reduces reuse-time candidate attempts from 16,570 to 144 while eliminating repeated historical composition expansion in the paired comparison. Combined with the cross-domain checks and the formal closure result, these findings support a simple architectural conclusion: recursive formation can produce a first-class operational object whose provenance remains available while its historical formation chain need not be reconstructed as part of later public search; verification then determines whether that object is retained for persistent reuse.

\appendix
\section{Experimental Record and Reproducibility Details}
\label{app:experimental-record}

\subsection{Search-Work Measures}

The paired depth experiment records work performed by the search and verification pipeline. The measures used throughout Section~\ref{sec:evaluation} are defined as follows.

\noindent\textbf{Candidate attempts.} Candidate search objects presented for consideration under the fixed search view.

\noindent\textbf{Binding combinations.} Slot-binding combinations considered by the search; these equal candidate attempts in the reported paired runs.

\noindent\textbf{Composition expansions.} Reconstruction events that expand historical composed nodes into the search process.

\noindent\textbf{Intermediate materializations.} Intermediate objects instantiated while reconstructing expanded composition routes.

\noindent\textbf{Verification batches.} Grouped verifier checks issued during the search procedure.

\noindent\textbf{D-level derivation depth.} Newly formed historical D-level chain steps in the search route, distinct from runtime call depth.

\subsection{Persistent-Store Accounting}

The sequence-domain experiments distinguish persistent records from executable capabilities. Stage I begins with 12 persistent records and adds 14 executable capabilities, yielding 26 persistent records in total: 24 executable capability records and two non-executable topology-template records. Stage II begins from those 26 records, adds 12 executable D-level capabilities, and ends with 38 persistent records, of which 36 are executable. The maximum formation level at the end of Stage II is ten (stored internally in the \texttt{canonical\_level} field).

\begin{table}[H]
\centering
\small
\begin{tabular}{@{}lrrr@{}}
\toprule
Snapshot & Persistent records & Executable capabilities & Auxiliary records \\
\midrule
Initial seed & 12 & -- & -- \\
After Stage I & 26 & 24 & 2 \\
After Stage II & 38 & 36 & 2 \\
\bottomrule
\end{tabular}
\caption{Persistent-record accounting for the sequence-domain acquisition stages.}
\label{tab:store-accounting}
\end{table}

\subsection{Reproducibility Record}

The public sequence-domain artifact, internally identified as V9.4, contains fixed goal definitions, seeds and digests, the starting-library snapshot, per-view traces, and acquisition and admission summaries. It is available at \url{https://github.com/lancelot138505-cpu/f-system}. Its package specification requires Python 3.10 or newer and NumPy 1.24 or newer. Sequence-domain results are also checked against a JavaScript executor. The paired comparison reuses the same goals, applicable examples and holdouts, oracle, seed, verifier, starting snapshot, and candidate limit in both search views; only the historical search representation changes between \emph{Expanded} and \emph{First-class}.

\begingroup
\raggedright
\Urlmuskip=0mu plus 1mu

\endgroup

\end{document}